\documentclass[1p,final]{elsarticle}
\usepackage{amsfonts,color,morefloats,pslatex}
\usepackage{amssymb,amsthm, amsmath,latexsym}
\usepackage[pagewise]{lineno}

\newtheorem{theorem}{Theorem}
\newtheorem{lemma}[theorem]{Lemma}

\newtheorem{corollary}[theorem]{Corollary}

\newtheorem{conj}{Conjecture}

\newcommand{\gf}{{\mathrm{GF}}}

\newcommand{\AG}{{\mathrm{AG}}}

\newcommand{\GA}{{\mathrm{GA}}}

\newcommand{\VS}{{\mathrm{VS}}}

\newcommand{\Stab}{{\mathrm{Stab}}}

\newcommand{\cP}{{\mathcal{P}}}
\newcommand{\cB}{{\mathcal{B}}}

\newcommand{\bD}{{\mathbb{D}}}

\begin{document}

\begin{frontmatter}



\title{Combinatorial $t$-designs from quadratic functions
\tnotetext[fn1]{The research of C. Xiang was supported by the National Natural Science Foundation of China (No.11701187) and the PhD Start-up Fund of the Natural Science Foundation of Guangdong Province of China (No.2017A030310522). The research of X. Ling was supported by National Natural Science Foundation of China (Grant No.
11871058).}
}

\author[cxiang]{Can Xiang}
\ead{cxiangcxiang@hotmail.com}
\author[xling]{Xin Ling}
\ead{xinlingcwnu@163.com}
\author[qwang]{Qi Wang}
\ead{wangqi@sustech.edu.cn}

\address[cxiang]{College of Mathematics and Informatics, South China Agricultural University, Guangzhou, Guangdong 510642, China}
\address[cmt]{School of Mathematics and Information, China West Normal University, Nanchong, Sichuan,  637002, China}
\address[qwang]{Department of Computer Science and Engineering, Southern University of Science and Technology, Shenzhen, Guangdong 518055,  China}

\begin{abstract}
Combinatorial $t$-designs have been an interesting topic in
combinatorics for decades. It was recently reported that the image sets of a fixed size of certain special polynomials may constitute a $t$-design. Till now only a small amount of
work on constructing $t$-designs from special polynomials has been done, and it is in general hard to determine their parameters. In this paper, we investigate this idea further by using quadratic functions over finite fields, thereby obtain infinite families of $2$-designs, and explicitly determine their parameters. The obtained designs cover some earlier $2$-designs as special cases. Furthermore, we confirmed Conjecture $3$ in Ding and Tang (arXiv: 1903.07375, 2019).

\end{abstract}


\end{frontmatter}

\section{Introduction}

Let $k$, $t$ and $v$ are positive integers with $1\leq t \leq k \leq v$. Let $\cP$ be a set of $v \ge 1$ mqi{$v$} elements, and let $\cB$ be a set of $k$-subsets of $\cP$. The pair $\bD = (\cP, \cB)$ is called an \emph{incidence structure}, and is said to be a $t$-$(v, k, \lambda)$ {\em design\index{design}} if every $t$-subset of $\cP$ is contained in exactly $\lambda$ elements of $\cB$.
The elements of $\cP$ are called points, and those of $\cB$ are referred to as blocks.
We usually use $b$ to denote the number of blocks in $\cB$.  A $t$-design is called {\em simple\index{simple}} if $\cB$ has no repeated blocks. A $t$-design is called symmetric if $v = b$ and trivial if $k = t$ or $k = v$. In this paper, we study only simple $t$-designs with $v > k > t$. When  $t \geq 2$ and $\lambda=1$, a $t$-design is called a {\em Steiner system\index{Steiner system}} and traditionally denoted by $S(t,k, v)$.

Combinatorial t-designs have found important applications in coding theory, cryptography, communications and statistics. There are two major methods of constructing $t$-designs.
%One of them is based on coding-theoretical. A number of $t$-designs have been constructed (see for example \cite{AK92,Ding2015,KP95,KP03,TDX,TD98,TD07}) by this method.
One is to construct them from error-correcting codes, and a number of constructions have been presented (see for example~\cite{AK92,Ding2015,KP95,KP03,TDX,TD98,TD07}).
Recently,  Ding and Li~\cite{DL17} obtained infinite families of $2$-designs and $3$-designs from some special codes and their duals.
Afterwards, some $t$-designs were further constructed from some other special codes over finite fields (see \cite{Ding20181,Ding20182,du1,du2,du3}).
The other method is via group actions of certain permutation groups which are $t$-transitive or $t$-homogeneous on a certain point set.
The following theorem, which shows that the incidence structure $(\mathcal P, \cB)$ is always a $2$-design by $2$-homogeneous group actions (see \cite[Proposition 4.6]{BJL}), was recently employed by Liu and Ding \cite{LD} to construct a number of infinite families of $2$-designs.

\begin{theorem}~\cite[Proposition 4.6]{BJL}\label{the-con}
Let $\mathcal P$ be a set of $v \geq 1$ elements, and let $G$ be a permutation group on $\mathcal P$. Let $B \subseteq \mathcal P$
be a $k$-subset with $k\geq 2$. Define
$$\cB=G(B) = \{g(B) : g \in G\},$$
where $g(B) = \{g(b) : b \in B\}$. If $G$ is $2$-homogeneous on $\mathcal P$ and $|B| \geq 2$, then $(\mathcal P, \cB)$ is a
$2-(v,k,\lambda)$ design with
$$
k=|B|,
\lambda=
b \left(
\begin{aligned}
&k  \\
&2  \\
\end{aligned}
\right)
/
\left(
\begin{aligned}
&v\\
&2\\
\end{aligned}
\right),
$$
where $b=\frac{|G|}{|G_{B}|}$ and $G_{B}=\{g\in G: g(B)=B\}$ is the stabilizer of $B$ under the group $G$.
\end{theorem}
Very recently, Ding and Tang \cite{DT192} presented two constructions of $t$-designs from special polynomials over finite fields, and obtained $2$-designs and $3$-designs with interesting parameters from their defined d-polynomials. However, it is in general hard to determine the parameters of the underlying $t$-designs by their constructions. Motivated by this fact, we obtain infinite families of $2$-designs by using quadratic functions over finite fields and determine their parameters explicitly. For other constructions of $t$-designs, see \cite{tang19,BJL,Shri,Mac,Reid} and the references therein.

%The
%objectives of this paper are to introduce two constructions of $t$-designs with special
%polynomials over finite fields $\gf(q)$, and obtain $2$-designs and $3$-designs with
%interesting parameters. A type of d-polynomials is defined and used to construct $2$-designs.
%Under the framework of the first construction, it is shown that
%every o-polynomial over $\gf(2^m)$ gives a $2$-design, and every o-monomial over
%$\gf(2^m)$ yields a $3$-design. Under the second construction, every $o$-polynomial
%gives a $3$-design. Some open problems and conjectures are also presented in this paper.

The rest of this paper is arranged as follows. Section \ref{sec-pre} introduces some basic notations and results of
projective planes and affine curves which will be needed in subsequent sections. Based on a generic construction in \cite{DT192}, Section \ref{sec-construct} gives infinite families of $2$-designs with new parameters by quadratic functions over finite fields and the proofs of their parameters
are given in Section \ref{sec-mainproof}. Section \ref{sec-summary} summarizes this paper.

\section{Preliminaries}\label{sec-pre}

In this section, we state some notations and basic facts on affine curves and projective planes, which will be used in the following sections.

\subsection{Some notations fixed throughout this paper}

For convenience, we adopt the following notations unless otherwise stated.
\begin{itemize}
  \item $p$ is a prime number.
 \item $\gcd(h_1,h_2)$  denotes the greatest common divisor of the two positive integers $h_1$ and $h_2$.
 \item $q = p^m$, where $m$ and $l$ are positive integers, and $\gcd(m, l) = 1$.
 \item  $\gf(q)$ denotes the finite field with $q$ elements and  $\gf(q)^{*} = \gf(q)\setminus \{0\}$.

  % \item $\overline{\gf(q)}$ denotes the algebraic closure of $\gf(q)$.
\item $\textup{QR}$ and $\textup{N\textup{QR}}$ denote the set of all nonzero quadratic residues and  quadratic non-residues in $\gf(q)$, respectively.
%\item $\eta$ and $\bar{\eta}$ are the quadratic characters of $\gf(q)^{*}$ and  $\gf(p)^{*}$, repsectively. We extend these quadratic characters by letting $\eta(0)=0$ and $\bar{\eta}(0)=0$.
\end{itemize}

\subsection{Projective planes and affine curves}

Let $\mathrm{GF}(q^{\infty})$ be the algebraic closure of $\mathrm{GF}(q)$. The projective plane $\mathbb P^2(
\mathrm{GF}(q))$ is defined as
$$\mathbb P^2(\mathrm{GF}(q)):=\left (\mathrm{GF}(q)^3\setminus \{(0,0,0)\} \right )/\thicksim ,$$
where $(X_0, Y_0, Z_0) \thicksim  (X_1, Y_1, Z_1)$ if and only if there is some $\lambda \in  \gf(q)^*$
with $X_1 =\lambda X_0,Y_1 =\lambda Y_0$ and $Z_1 =\lambda Z_0$. To remind ourselves that points of $\mathbb P^2(\mathrm{GF}(q))$ are equivalence classes,
we write $(X: Y: Z)$ for the equivalence class of $(X, Y, Z)$ in $\mathbb P^2(\mathrm{GF}(q))$ .
Let $f(x,y) \in  \mathrm{GF}(q)[x,y]$ be a polynomial of
degree $d$ over $\mathrm{GF}(q) $. Then the affine curve $C_f$ associated to $f$ is defined by
$$C_f=\{(x,y)\in  \mathrm{GF}(q^{\infty})^2: f(x,y)=0 \}.$$
The projective closure of
the affine curve $C_f$ is
$$\hat{C}_f=\{(X:Y:Z)\in \mathbb P^2(\mathrm{GF}(q^{\infty})): F(X,Y,Z)=0 \},$$
where $F(X,Y,Z)=Z^d \cdot  f \left(\frac{X}{Z},\frac{Y}{Z} \right )$
is the homogenization  of $f$. For polynomial $F$, $F_X$, $F_Y$ and $F_Z$ denote the formal partial derivatives  of $F$ with respect to $X$, $Y$ and $Z$, respectively.
A singular point of $\hat{C}_f$ is a point $(X_0:Y_0:Z_0)\in \mathbb P^2(\mathrm{GF}(q^{\infty}))$ such that
$$
\left\{ \begin{array}{l}
F(X_0,Y_0,Z_0)=0 \\
F_X(X_0,Y_0,Z_0)=0 \\
F_Y(X_0,Y_0,Z_0)=0       \\
F_Z(X_0,Y_0,Z_0)=0.
\end{array}
\right.
$$
The projective curve $\hat{C}_f$ is nonsingular if it has no singular points.
A nonsingular projective plane curve is  irreducible.

Let $\mathfrak X$ be a curve over $\gf(q)$, whose defining equations have coefficients in $\mathrm{GF}(q)$.
Then the points on $\mathfrak  X$ with all their coordinates in $\mathrm{GF}(q)$ are called $\mathrm{GF}(q)$-rational points.
The set of all $\mathrm{GF}(q)$-rational points of $\mathfrak  X$ is denoted by $\mathfrak  X \left ( \mathrm{GF}(q) \right )$.

The following theorem is the fundamental result in the area of algebraic curves.
\begin{theorem}[Hasse-Weil Theorem]\label{thm:Hasse-Weil}
Let $\mathfrak  X$ be a nonsingular projective curve of genus $g$ over the field $\mathrm{GF}(q)$ and set $N = | \mathfrak  X \left ( \mathrm{GF}(q) \right )|$. Then
\begin{align}\label{eq:gg}
| N-(q+1)| \leq  2g \sqrt{q}.
\end{align}

\end{theorem}
If $q$ is not a perfect square, we can  replace the right-hand side of the inequality (\ref{eq:gg}) in Hasse-Weil Theorem  with $g \lfloor 2 \sqrt{q} \rfloor$.

If $\hat{C}_f$ is a nonsingular projective plane curve corresponding to the polynomial $f(x,y)\in \mathrm{GF}(q)[x,y]$ of degree $d$,
then the genus $g$ of $\hat{C}_f$  is given by the Pl\"ucker formula
\begin{align}\label{eq:genus}
g=\frac{(d-1)(d-2)}{2}.
\end{align}

\section{$t$-designs from quadratic functions over $\gf(q)$}\label{sec-construct}

%Throughout this paper, let $p$ be a prime number and $q=p^m$, where $m$ be a positive integer.
Let $f$ be a polynomial over $\gf(q)$, which is always viewed
as a function from $\gf(q)$ to $\gf(q)$. For each $(b, c) \in \gf(q)^2$,
define
\begin{eqnarray}
B_{(f,b,c)}=\{f(x)+bx+c:  x \in \gf(q)\}.
\end{eqnarray}
Let $k$ be an integer with $2 \leq k \leq q$. Define
\begin{eqnarray}
  \cB_{(f,k)}=\{B_{(f,b,c)}: |B_{(f,b,c)}|=k, \textrm{ and } b, \ c \in \gf(q)\}.
\end{eqnarray}
The incidence structure $\bD(f, k):=(\gf(q), \cB_{(f,k)})$ may be a $t$-$(q, k, \lambda)$ design for some $\lambda$,
where $\gf(q)$ is the point set, and the incidence relation is given by the set membership. In such a
case, we say that the polynomial $f$ supports a $t$-$(q, k, \lambda)$ design. This construction of $t$-designs with polynomials over finite fields was documented recently in \cite{DT192}.

We define the \emph{value spectrum} of a polynomial over $\gf(q)$ to be the multiset
\begin{eqnarray*}
\VS(f)=\{\{ |B_{(f,b,c)}|: (b,c) \in \gf(q)^2 \}\}.
\end{eqnarray*}
To determine the parameters of $t$-designs supported by a polynomial $f$, we need to
know its value spectrum.

This construction is generic in the sense that $t$-designs could be produced by properly
selecting the polynomial $f$ over GF(q).
%If the polynomial $f$ is well chosen, some $t$-designs with new parameters can be obtained.
Based on this construction, only a small number of $t$-designs have been constructed.
% since Ding and Tang published their paper in 2019 \cite{DT192} .
One of the main reasons is that the value spectrum of a polynomial is hard
to determine in general. In this paper, we consider constructing $t$-designs from the quadratic function
\begin{eqnarray*}\label{eq-f}
f(x)=x^{p^l+1}
\end{eqnarray*}
over $\gf(q)$ and determine their parameters.

The following two theorems are the main results of this paper, whose proofs will be postponed to present in Section 4.

\begin{theorem}\label{thm-main1}
Let $p=2$ ,$\ell, m$ be two positive integers with $m\ge 3$, $\ell < \frac{m}{4}-1$ and $\gcd(\ell, m)=1$. Let $f(x)=x^{2^{\ell}+1}$. Then the incidence structure $\bD(f(x), k):=(\gf(q), \cB_{(f(x),k)})$ is a $2$-$(q, k, k(k-1))$ design, where $k=\frac{2q+(-1)^m}{3}$.
\end{theorem}

\begin{theorem}\label{thm-main2}
Let $p$ be an odd prime with $p \equiv 3 \pmod 4$ and $m\ge 3$ be odd. Let $\ell$ be a positive integer with  $\ell < \frac{m-2}{4}$ and $\gcd(\ell, m)=1$. Let
$f(x)=x^{p^{\ell}+1}$.  Then the incidence structure $\bD(f(x), k):=(\gf(q), \cB_{(f(x),k)})$ is a $2$-$(q, k, \frac{k(k-1)}{2})$ design, where $k=q-\frac{pq-1}{2(p+1)}$.
\end{theorem}

As a special case of Theorem \ref{thm-main2}, we have the following corollary.
\begin{corollary}\label{coro-caix}
Let $(p,\ell)=(3,2)$ and $m\geq 11 $ be odd. Then the incidence structure $\bD(x^{10}, k):=(\gf(3^m), \cB_{(x^{10},k)})$ is a $2$-$(3^m, k, \frac{k(k-1)}{2})$ design, where $k=\frac{5 \cdot 3^m+1}{8}$.
\end{corollary}
Note that the conclusion of Corollary \ref{coro-caix} also follows if $m\in \{3,5,7,9\}$, which is verified by the Magma program. This means that the conjecture $3$ in Ding and Tang \cite{DT192} is true.

%\begin{eqnarray}\label{eqn-miu}
%\mu&=&|\Stab(J_{\ell})|     \\                                     \nonumber
%&=&\left\{ \begin{array}{ll}
%\dot.....                                           & \mbox{ if $m$ is even}, \\
%\dot.....                  & \mbox{ if $m$ is odd and $p^{\ell}$ is odd}, \\
%\dot.....                 & \mbox{ if $m$ is odd and $p^{\ell}$ is even}.
%\end{array}
%\right.
%\end{eqnarray}
\section{Proofs of the main results}\label{sec-mainproof}

Our task of this section is to prove Theorems \ref{thm-main1} and \ref{thm-main2}. To this end, we shall
prove a few more auxiliary results before proving the main results of this paper.

\subsection{Some auxiliary results}\label{sec-sub1}

\begin{lemma}\label{lem-airr}
Let $\alpha \in \gf(q)^{*}$ and $\beta \in \gf(q)^{*}$. Let $\ell$ and $m$ be integers with $1 \le \ell < m$. Let
$f=f(x,y)=x^{p^{\ell}+1}+x-\alpha \left (y^{p^{\ell}+1}+y \right )-\beta \in \gf(q)[x,y]$ and $N=|C_f|$.
Then
\[(q+1- \delta)-p^{\ell}(p^{\ell}-1)\sqrt{q}  \le N \le (q+1)+p^{\ell}(p^{\ell}-1)\sqrt{q} ,\]
where $\delta=\gcd(p^\ell+1, p^m-1)$.
\end{lemma}
\begin{proof}
Let $\mathfrak{X}$ be the projective curve $\hat{C}_f$. Let
$$F(X,Y,Z)=X^{p^{\ell}+1}+XZ^{p^{\ell}}-\alpha \left (Y^{p^{\ell}+1}+YZ^{p^{\ell}} \right )-\beta Z^{p^{\ell}+1} \in \mathrm{GF}(q)[X,Y,Z]$$ be the homogenization of $f(x,y)$ and
$(X: Y: Z)$ be a singular point of $\mathfrak{X}$. Then we have
\begin{equation*}
\left\{
\begin{aligned}
&F_X=X^{p^{\ell}}+Z^{p^{\ell}}=0\\
&F_Y=\alpha \left ( Y^{p^{\ell}}+Z^{p^{\ell}} \right )=0\\
&F_Z= \beta Z^{p^{\ell}} =0\\
&F(X,Y,Z)=0\\
\end{aligned}
\right.
\end{equation*}
Thus,
\begin{equation*}
\left\{
\begin{aligned}
&X^{p^{\ell}}+Z^{p^{\ell}}=0\\
&Y^{p^{\ell}}+Z^{p^{\ell}}=0\\
&\beta Z^{p^{\ell}} =0\\
\end{aligned}
\right.
\end{equation*}
From $\beta \neq 0$,  it follows that $X = Y = Z = 0$, a contradiction. Thus, $\mathfrak{X}$
is a nonsingular projective curve. By the Pl\"ucker formula (\ref{eq:genus}) and Theorem \ref{thm:Hasse-Weil}, we have
\begin{align}\label{eq-proj-num}
(q+1)-p^{\ell}(p^{\ell}-1)\sqrt{q}  \le \mathfrak{X}\left ( \mathrm{GF}(q) \right ) \le (q+1)+p^{\ell}(p^{\ell}-1)\sqrt{q}.
\end{align}
By multiplying through by a nonzero element of $\gf(q)$,
 we can assume the right-most nonzero coordinate of a point of $\mathbb{P}^2\left (\mathrm{GF}(q) \right )$ is $1$.
Therefore, we have
\begin{align*}
\mathfrak{X}\left (\mathrm{GF}(q) \right )= \left \{ (x:y:1): (x,y) \in C_f \right \} \cup S,
\end{align*}
where $S=\{(x:1:0) : x\in \mathrm{GF}(q), x^{p^\ell+1}=\alpha\}$.
Then
$$
| \mathfrak{X} (\mathrm{GF}(q) |=N+| S|.
$$
Since $| S| \le \gcd(p^\ell+1, p^m-1)$, the desired results follows from Inequality (\ref{eq-proj-num}).
\end{proof}

\begin{lemma}\label{lem-proj-(q+1)}
Let $\alpha \in \gf(q)\setminus \{0, 1\}$
, $\ell$ and $m$ be integers with $1 \le \ell < m$. Let
$f(x,y)=x^{p^{\ell}+1}+x-\alpha \left (y^{p^{\ell}+1}+y \right ) \in \mathrm{GF}(q)[x,y]$ and $N=| C_f ( \gf(q)  ) |$.
Then
\[  q+1-\delta \le N \le q+1,\]
where $\delta=\gcd(p^\ell+1, p^m-1)$.
\end{lemma}

\begin{proof}
Let $\mathfrak{X}$ be the projective curve $\hat{C}_f$ and
$F(X,Y,Z)=X^{p^{\ell}+1}+XZ^{p^{\ell}}-\alpha \left (Y^{p^{\ell}+1}+YZ^{p^{\ell}} \right ) \in \mathrm{GF}(q)[X,Y,Z]$ be the homogenization of $f(x,y)$.

Let $(X:Y: Z)  \in \mathfrak{X}\left (  \mathrm{GF}(q)\right )$. Then we have
\[\left ( X-\alpha Y \right ) Z^{p^\ell}= \alpha Y^{p^\ell+1} - X^{p^\ell +1}.\]
If $X=\alpha Y$, then $0=\alpha Y^{p^\ell+1}- \alpha^{p^\ell +1} Y^{p^\ell+1}=\alpha \left (1-\alpha^{p^\ell}  \right ) Y^{p^\ell+1} $.
By $\alpha\neq 0$ and $\alpha\neq 1 $, we know that $(X:Y:Z)$ must be the point $(0: 0: 1)$.\\
If $X\neq \alpha Y$ and $Y=0$, then $X Z^{p^\ell}= - X^{p^\ell +1}$. Thus,  $(X:Y:Z)$ must be the point $(1: 0: -1)$.\\
If $X\neq \alpha Y$ and $Y\neq 0$, then
\[ Z=\left ( \frac{\alpha Y^{p^\ell+1}-X^{p^\ell+1}}{X-\alpha Y} \right )^{p^{m-\ell}}.\]
Thus,  $(X:Y:Z)$ must be the point $ \left (x: 1: \left ( \frac{\alpha -x^{p^\ell+1}}{x-\alpha } \right )^{p^{m-\ell}} \right )$ with $x\in \mathrm{GF}(q) \setminus \{ \alpha \}$.
Hence,
\begin{align}\label{eq-proj-num-(q+1)}
  \mathfrak{X}\left ( \mathrm{GF}(q) \right ) =  q+1.
\end{align}
Note that
\begin{align*}
\mathfrak{X}\left (\mathrm{GF}(q) \right )= \left \{ (x:y:1): (x,y) \in C_f \right \} \cup S,
\end{align*}
where $S=\{(x:1:0) : x\in \mathrm{GF}(q), x^{p^\ell+1}=\alpha\}$.
It then follows that
$$|\mathfrak{X}\left (\mathrm{GF}(q) \right )|=N+ |S|.$$
Since $|S|\leq \gcd(p^\ell+1, p^m-1)$, the proof is then completed by Inequality (\ref{eq-proj-num-(q+1)}).
\end{proof}

By Lemmas \ref{lem-airr} and \ref{lem-proj-(q+1)}, we have the following corollary.
\begin{corollary}\label{coro-N11}
Let $(\alpha, \beta) \in \mathrm{GF}(q)^{*} \times \mathrm{GF}(q)$  with $(\alpha, \beta) \neq (1,0)$
. Let $\ell$ and $m$ be integers with $1 \le \ell < m$. Let $N=| C_f \left ( \gf(q) \right ) |$ , where
$f(x,y)=x^{p^{\ell}+1}+x-\alpha \left (y^{p^{\ell}+1}+y \right )-\beta \in \mathrm{GF}(q)[x,y]$.
Then
\[(q- p^{\ell})-p^{\ell}(p^{\ell}-1)\sqrt{q}  \le N \le (q+1)+p^{\ell}(p^{\ell}-1)\sqrt{q}.\]
\end{corollary}

\begin{lemma}\label{lem-N11}
Let $(\alpha, \beta) \in \mathrm{GF}(q)^{*} \times \mathrm{GF}(q)$  with $(\alpha, \beta) \neq (1,0)$.
Let $\ell$ and $m$ be integers with $1 \le \ell < m$.
Let $N=| C_f \left ( \mathrm{GF}(q) \right )| $, where
$f(x,y)=x^{p^{\ell}+1}+x-\alpha \left (y^{p^{\ell}+1}+y \right )-\beta \in \mathrm{GF}(q)[x,y]$.
Define
\begin{eqnarray}\label{eqn-Je11}
B_{\ell}= \left \{ x^{p^\ell+1}+x:  x\in \gf (q)\right \}.
\end{eqnarray}
If $B_{\ell}= \alpha B_l +\beta $,
then
\[N \ge 2q- | B_{\ell}|.\]
\end{lemma}
\begin{proof}
Let $h_0(x)=x^{p^\ell+1}+x$ and $h_1(x)=\alpha \left (x^{p^{\ell}+1}+x \right )+\beta$.
Let $k= | B_{\ell}|$ and $B_{\ell}=\{z_1, \cdots, z_k\}$. For any $z\in \gf(q)$,
let $h_0^{-1}(z)=\{x\in \gf(q): h_0(x)=z\}$.
Then we have
\[ N=\sum_{y\in \gf (q)} | h_0^{-1} \left (h_1(y) \right )|.\]
Since $B_{\ell}= \alpha B_l +\beta $, we have  $| h_0^{-1} \left (h_1(y) \right ) |\ge 1$ for any $y\in \gf(q)$.
Then
\begin{align*}
N& =q+\sum_{y\in \gf (q)} \left ( | h_0^{-1} \left (h_1(y) \right )| -1 \right )\\
&\ge q+ \sum_{i=1}^k \left ( | h_0^{-1} \left (z_i \right )| -1 \right )\\
& = q-k + \sum_{i=1}^k | h_0^{-1} \left (z_i \right )| \\
&= 2q-k.
\end{align*}
This then completes the proof.
\end{proof}

In order to obtain Corollary \ref{coro-N2}, we need the following result which was proved in \cite[Theorem 5.6]{Blu}.

\begin{lemma}~\cite[Theorem 5.6]{Blu}\label{lem-N0}
Let $F$ be an arbitrary finite field of characteristic $p$, $s$ be a power of $p$ and $F \bigcap \gf(s)=\gf(t)$. Let $0\neq b \in F$ and $N_0$ denote the number of $b$ such that the polynomial
$x^{s+1}-bx+b$ has no rational root in $F$. Then

\begin{eqnarray}\label{eqn-N0}
N_0 &=&\left\{ \begin{array}{ll}
\frac{t^{\hat{m}+1}-t}{2(t+1)}    & \mbox{ if $\hat{m}$ is even}, \\
\frac{t^{\hat{m}+1}-1}{2(t+1)}    & \mbox{ if $\hat{m}$ is odd and $s$ is odd}, \\
\frac{t^{\hat{m}+1}+t}{2(t+1)}    & \mbox{ if $\hat{m}$ is odd and $s$ is even},
\end{array}
\right.
\end{eqnarray}
where $\hat{m}=[F:\gf(t)]$.
\end{lemma}

\begin{corollary}\label{coro-N2}
Let $\ell$ be a positive integer with $\gcd(\ell, m)=1$. Let $\hat{N}$ denote the number of $c\in \gf(q)^*$ such that the polynomial
$x^{p^l+1}+x+c$ has no rational root in $\gf(q)$. Then

\begin{eqnarray}\label{eqn-N2}
\hat{N} &=&\left\{ \begin{array}{ll}
\frac{p^{m+1}-p}{2(p+1)}    & \mbox{ if $m$ is even}, \\
\frac{p^{m+1}-1}{2(p+1)}    & \mbox{ if $m$ is odd and $p^{\ell}$ is odd}, \\
\frac{p^{m+1}+p}{2(p+1)}    & \mbox{ if $m$ is odd and $p^{\ell}$ is even}.
\end{array}
\right.
\end{eqnarray}
\end{corollary}
\begin{proof}
  In Lemma \ref{lem-N0}, we let $F=\gf(p^m)$ and $s=p^{\ell}$ with $\gcd(\ell, m)=1$. Then $$F \cap \gf(s)= \gf(p^m)\cap \gf(p^{\ell})=\gf(p)$$
and
$$\hat{m}=[F:\gf(t)]=[\gf(p^m):\gf(p)]=m.$$
Further, since $x^{p^{\ell}}$ is a permutation of $\gf(p^m)$, we have
\begin{align*}
x^{s+1}-bx+b& =x^{p^{\ell}+1}-bx+b           \\
&= x^{p^{\ell}+1}-b^{p^{\ell}}x+b^{p^{\ell}}~~~~~~~~~~~~~~~~~~~~~~~~~~~~~~~~~(b~~is ~~replaced~~with~~{b^{\ell}}) \\
&= (-by)^{p^{\ell}+1}-b^{p^{\ell}}(-by)+b^{p^{\ell}}~~~~~~~~~~~~~~~~~(Let~~x=-by) \\
&= b^{{p^{\ell}+1}}y^{p^{\ell}+1}+b^{p^{\ell}+1}y+b^{p^{\ell}} ~~~~~~~~~~~~~~~~~~~~~~~~~~~~~~~~~~~\\
&=b^{{p^{\ell}+1}}(y^{p^{\ell}+1}+y+b^{-1}) .
\end{align*}
Since $b\in \gf(p^m)^*$, $b^{{p^{\ell}+1}}(y^{p^{\ell}+1}+y+b^{-1})=0$ is equivalent to
$$
x^{p^{\ell}+1}+x+c=0,
$$
where $c\in \gf(p^m)^*$. The desired conclusion then follows from Lemma \ref{lem-N0}.
\end{proof}

\begin{lemma}\label{lem-kk1}
Let $m$ and $\ell$ be a positive integer with $\gcd(\ell, m)=1$.
Then
\begin{eqnarray}\label{eqn-k111}
|B_{\ell}|
 &=&\left\{ \begin{array}{ll}
q-\frac{p^{m+1}-p}{2(p+1)}    & \mbox{ if $m$ is even}, \\
q-\frac{p^{m+1}-1}{2(p+1)}    & \mbox{ if $m$ is odd and $p^{\ell}$ is odd}, \\
q-\frac{p^{m+1}+p}{2(p+1)}    & \mbox{ if $m$ is odd and $p^{\ell}$ is even},
\end{array}
\right.
\end{eqnarray}
where $B_\ell$ was defined by (\ref {eqn-Je11}).
\end{lemma}

\begin{proof}
By definition, we have
$$|B_{\ell}|=q-\hat{N}$$
where $\hat{N}$ was defined by Corollary \ref{coro-N2}. This means that Equation (\ref{eqn-k111}) follows. This completes the proof.
\end{proof}

\begin{lemma}\label{lem-miu2}
Let $m$ and $\ell$ be a positive integer with $\gcd(\ell, m)=1$. Define
$$
\Stab(B_{\ell})=\{ux+v: (u,v) \in \gf(q)^* \times \gf(q), \ uB_{\ell}+v=B_{\ell}\}
$$
and $\mu=|\Stab(B_{\ell})|$, where $B_\ell$ was defined by (\ref {eqn-Je11}).
Then we have the following.

(\uppercase\expandafter{\romannumeral1}) If $p=2$, $m\geq 4 $ is even and $2\ell+2< m/2$, then $\mu=1$.

(\uppercase\expandafter{\romannumeral2}) If $p=2$, $m\geq 3 $ is odd and $2\ell+2< m/2$, then $\mu=1$.

(\uppercase\expandafter{\romannumeral2}) If $p\equiv3 ~mod~4$, $m\geq 3 $ is odd and $2\ell+1< m/2$, then $\mu=1$.

\end{lemma}

\begin{proof}
%Let us give the proofs of three cases, respectively.
We now prove the three cases in the following.

(\uppercase\expandafter{\romannumeral1})
By definition, it is clear that $(1,0)\in \Stab(B_{\ell})$. Suppose that $\mu=|\Stab(B_{\ell})| \neq 1$, then there must exist
$(\alpha, \beta) \in \Stab(B_{\ell}) $  with $(\alpha, \beta) \neq (1,0)$. From  Corollary \ref{coro-N11}, it follows that
\begin{eqnarray}\label{eqn-md1}
N\le (q+1)+p^{\ell}(p^{\ell}-1)\sqrt{q},
\end{eqnarray}
where $N$ was defined by Corollary \ref{coro-N11}.

Meanwhile, by Lemmas \ref{lem-N11} and \ref{lem-kk1}, we have
\begin{eqnarray}\label{eqn-md2}
N\geq 2q-k=q+\frac{p^{m+1}-p}{2(p+1)}=(q+1)+ \frac{p^{m+1}-3p-2}{2(p+1)}.
\end{eqnarray}
Since $p=2$ and $m\geq4$ is even, we have
$$
\frac{p^{m+1}-3p-2}{2(p+1)}-2^{m-2}=\frac{2^{m+1}-8}{6}- 2^{m-2}=\frac{1}{3}(2^{m-2}-4)\geq 0.
$$
This means that
\begin{eqnarray}\label{eqn-md3}
\frac{p^{m+1}-3p-2}{2(p+1)}\geq2^{m-2}.
\end{eqnarray}
Therefore, from  Equations (\ref{eqn-md2}) and (\ref{eqn-md3}), we have
\begin{eqnarray}\label{eqn-md4}
N\geq(q+1)+2^{m-2}.
\end{eqnarray}
Furthermore, by $2\ell+2< m/2$ and Equations (\ref{eqn-md1}) we have
\begin{eqnarray}\label{eqn-md5}
N\le (q+1)+2^{\ell}(2^{\ell}-1)2^{m/2}< (q+1)+2^{2\ell+m/2}< (q+1)+2^{m-2},
\end{eqnarray}
which contradicts to Equations (\ref{eqn-md4}). This means that there does not exist $(\alpha, \beta) \in \Stab(B_{\ell})$  with $(\alpha, \beta) \neq (1,0)$. Hence, $\mu=|\Stab(B_{\ell})| =|\{(1,0)\}|= 1$.

(\uppercase\expandafter{\romannumeral2})
The proof is similar to case (\uppercase\expandafter{\romannumeral1}) and we omit it here. The desired conclusion then follows from
Lemma \ref{lem-N11} and  Corollary \ref{coro-N11}.

(\uppercase\expandafter{\romannumeral3}) By definition, it is clear that $(1,0)\in \Stab(B_{\ell})$. Suppose that $\mu=|\Stab(B_{\ell})| \neq 1$, then there must exist
$(\alpha, \beta) \in \Stab(B_{\ell})$  with $(\alpha, \beta) \neq (1,0)$. From  Corollary \ref{coro-N11}, we have

\begin{eqnarray}\label{eqn-md21}
N\le (q+1)+p^{\ell}(p^{\ell}-1)\sqrt{q},
\end{eqnarray}
where $N$ is defined by Corollary \ref{coro-N11}. Meanwhile, by Lemmas \ref{lem-N11} and \ref{lem-kk1}, we have
\begin{eqnarray}\label{eqn-md22}
N\geq 2q-k=q+\frac{p^{m+1}-1}{2(p+1)}=(q+1)+ \frac{p^{m+1}-2p-3}{2(p+1)}.
\end{eqnarray}
Since $p\geq 3$ and $m\geq3$ is odd, we have
$$
(p^{m+1}-2p-3)-2(p+1)p^{m-1}=p(p^{m-2}(p^2-2p-2)-2)-3 \geq 0.
$$
This means that
\begin{eqnarray}\label{eqn-md23}
\frac{p^{m+1}-2p-3}{2(p+1)}\geq p^{m-1}.
\end{eqnarray}
Therefore, from  Equations (\ref{eqn-md22}) and (\ref{eqn-md23}), we have
\begin{eqnarray}\label{eqn-md24}
N\geq(q+1)+p^{m-1}.
\end{eqnarray}
Further, by $2\ell+1< m/2$ and Equations (\ref{eqn-md21}) we have
\begin{eqnarray}\label{eqn-md5}
N\le (q+1)+p^{\ell}(p^{\ell}-1)p^{m/2}< (q+1)+p^{2\ell+m/2}< (q+1)+p^{m-1},
\end{eqnarray}
which is a contradiction to Equation (\ref{eqn-md24}). This means that there does not exist $(\alpha, \beta) \in \Stab(B_{\ell})$  with $(\alpha, \beta) \neq (1,0)$. Hence, $\mu=|\Stab(B_{\ell})| =|\{(1,0)\}|= 1$.

This completes the proof.
\end{proof}

\begin{lemma}\label{lem-qici}
Let $p\geq3$ and $p\equiv3~mod~4$. Let $m$ be odd and $\ell$ be a positive integer with $\gcd(\ell, m)=1$. Define the group
\begin{eqnarray}\label{eqn-Ag122}
\GA_1(\gf(q))=\{ux+v: (u, v) \in \gf(q)^* \times \gf(q), u \in \textup{QR} \}.
\end{eqnarray}
Then the group $\GA_1(\gf(q))$ is $2$-homogeneous on $\gf(q)$.
\end{lemma}

\begin{proof}
Let $\{x_1,x_2\}\subseteq \gf(q)$ and $\{y_1,y_2\}\subseteq \gf(q)$
be any two 2-subsets of $\gf(q)$. Let
\begin{eqnarray*}
\left\{ \begin{array}{l}
u_1x_1+v_1=y_1    \\
u_1x_2+v_1=y_2
\end{array}
\right.
\mbox{or}~
\left\{ \begin{array}{l}
u_1x_1+v_1=y_2    \\
u_1x_2+v_1=y_1.
\end{array}
\right.
\end{eqnarray*}
Then we have
\begin{eqnarray}\label{eqn-qici2}
\left\{ \begin{array}{l}
u_1=(x_1-x_2)^{-1}(y_1-y_2)    \\
v_1=y_1-(x_1-x_2)^{-1}(y_1-y_2)x_1
\end{array}
\right.
\mbox{or}~
\left\{ \begin{array}{l}
u_1=(x_1-x_2)^{-1}(y_1-y_2)(-1)    \\
v_1=y_2-(x_1-x_2)^{-1}(y_1-y_2)(-1)x_1.
\end{array}
\right.
\end{eqnarray}
By assumption, we have $-1\in \textup{NQR}$. It then deduce that one is a quadratic residue and the other is a quadratic non-residue in the two values $(x_1-x_2)^{-1}(y_1-y_2)$ and $(x_1-x_2)^{-1}(y_1-y_2)(-1)$ of Equation (\ref{eqn-qici2}). This means that there exists $\sigma(x)=(ux+v)\in \GA_1(\gf(q))$ such that $\sigma(x)$ sends $\{x_1,x_2\}$ to $\{y_1,y_2\}$, where $u\in \textup{QR}$ is a quadratic residue and $v\in \gf(q)$.
The desired conclusion then follows from the definition of $2$-homogeneity.
\end{proof}
\begin{lemma}\label{lem-dengzhi}
With the symbols and notation above, let
$$A_1=\{B_{(f,b,c)}:(b, c) \in \gf(q)^*\times \gf(q)\},$$
$$A_2=\{uB_{\ell}+v:(u, v) \in \gf(q)^*\times \gf(q)\}$$
and
$$A_3=\{uB_{\ell}+v:(u, v) \in \gf(q)^*\times \gf(q)~and~u\in \textup{QR}\},$$
where $f(x)=x^{p^{\ell}+1}$ and $B_{\ell}$ was defined by (\ref{eqn-Je11}).
Then we have the following results.

(\uppercase\expandafter{\romannumeral1}) If $p=2$ and $\gcd(\ell, m)=1$, then $A_1=A_2$.

(\uppercase\expandafter{\romannumeral2}) If $p\geq3$, $p\equiv3~mod~4$, $m$ is odd with $\gcd(\ell, m)=1$, then $A_1=A_3$.
\end{lemma}
\begin{proof}
For each $(b, c) \in \gf(q)^*\times \gf(q)$, we have
\begin{eqnarray}\label{eqn-aa}
f(x)+bx+c& =&x^{p^{\ell}+1}+bx+c          \nonumber \\
& =& x^{p^{\ell}+1}+b^{p^{\ell}}x+c~~~~~~~~~~~~~~~~~~~~~~~~~~    \mbox(b~is~replaced~with ~{b^{p^{\ell}}})        \nonumber      \\
& =&  (bx)^{p^{\ell}+1}+b^{p^{\ell}}(bx)+c~~~~~~~~~~~~~~~(x~is~replaced~with ~bx)                        \nonumber \\
& =&  b^{{p^{\ell}+1}}(x^{p^{\ell}+1}+x)+c. ~~~~~~~~~~~~~~~~~~~~~~~~~~~~~~~~~~~
\end{eqnarray}

(\uppercase\expandafter{\romannumeral1}) For any $B_{(f,b,c)}=\{f(x)+bx+c: x\in \gf(q)\}\in A_1$, from Equation (\ref{eqn-aa}) we have
$$B_{(f,b,c)}=\{f(x)+bx+c: x\in \gf(q)\}=\{b^{{p^{\ell}+1}}B_{\ell}+c : x\in \gf(q)\}\in A_2,$$
which means that $A_1\subseteq A_2$.  Next we prove $A_2\subseteq A_1$.

For each $(u, v) \in \gf(q)^*\times \gf(q)$, we have
\begin{eqnarray}\label{eqn-aa2}
u(x^{p^{\ell}+1}+x)+v& =&(u^{(p^{\ell}+1)^{-1}}x)^{p^{\ell}+1}+u^{1-(p^{\ell}+1)^{-1}}(u^{(p^{\ell}+1)^{-1}}x)+v          \nonumber \\
& =& (ux)^{p^{\ell}+1}+ux+v~~~~~~~~~~~~~~~~~~~~~~\mbox(since~p=2~,(p^{\ell}+1)^{-1}=1 )        \nonumber      \\
& =&  x^{p^{\ell}+1}+1\cdot x+v~~~~~~~~~~~~~~~~~~~~~~~~~~(ux~is~replaced~with ~x)
\end{eqnarray}
Hence, for any $uB_{\ell}+v\in A_2$, by Equation (\ref{eqn-aa2}) we have
$uB_{\ell}+v=B_{(f,1,v)}\in A_1.$
This means that $A_2\subseteq A_1$. The desired conclusion then follows.

(\uppercase\expandafter{\romannumeral2})
By definition, $\gcd(p^{\ell}+1,q-1)=2$. Thus, $b^{p^{\ell}+1}=(b^{(p^{\ell}+1)/2})^2\in \textup{QR}$ in Equation (\ref{eqn-aa}), which means that $A_1\subseteq A_3$.  The proof of $A_3\subseteq A_1$ is similar to the proof of $A_2\subseteq A_1$ of case (\uppercase\expandafter{\romannumeral1}) and we omit it here.
\end{proof}

%\begin{lemma}\label{lem-j1}
%With the symbols and notation above in Theorem \ref{thm-main}, we have
%$$k=|J_l|=\dot.....$$
%\end{lemma}

\subsection{The proofs of Theorems \ref{thm-main1} and \ref{thm-main2}}\label{sec-sub2}

It is now time to show the results as stated in Theorems \ref{thm-main1} and \ref{thm-main2}.

%(\uppercase\expandafter{\romannumeral1})  The proof of Theorem \ref{thm-main1} is as follows :
\begin{proof}[Proof of Theorem \ref{thm-main1}]
Recall that $p=2$ and $f(x)=x^{p^\ell+1}=x^{2^\ell+1}$.
By definition, from Lemma \ref{lem-kk1}, it follows that
\begin{eqnarray}\label{eqn-kkkk1}
k=|B_{\ell}|=|\{x^{p^l+1}+x: x \in \gf(q)\}|=\frac{2q+(-1)^m}{3}.
\end{eqnarray}
Define the group
\begin{eqnarray}\label{eqn-Ag11}
\GA(\gf(q))=\{ux+v: (u, v) \in \gf(q)^* \times \gf(q)\}.
\end{eqnarray}
It is clear that $\GA(\gf(q))$ is the general affine group and its size is $|\GA(\gf(q)|=q(q-1)$.
The stabilizer of $B_{\ell}$ under $\AG(\gf(q))$ is defined by
$$
\Stab(B_{\ell})=\{ux+v: (u,v) \in \gf(q)^* \times \gf(q), uB_{\ell}+v=B_{\ell}\},
$$
where $B_{\ell}$ is defined by (\ref{eqn-Je11}).
We then deduce that
\begin{equation}\label{eqn-xx}
|\cB_{(f,k)}|=\frac{|\GA(\gf(q)|}{|\Stab(B_{\ell})|} =\frac{(q-1)q}{|\Stab(B_{\ell})|}=q(q-1)~~~~~~~~~~~%(by~ Lemma ~\ref{lem-miu2}).
\end{equation}
by Lemma \ref{lem-miu2}.
This means that all blocks $B_{(f,b,c)}$
with $(b, c) \in \gf(q)^* \times \gf(q)$ are pairwise distinct. Note that $\GA(\gf(q))$ is $2$-homogeneous on $\gf(q)$. By definitions and the result (\uppercase\expandafter{\romannumeral1}) of Lemma \ref{lem-dengzhi}, the incidence structure $(\gf(q), \cB_{(f(x),k)})$ can be seen as $(\gf(q), \cB)$, which is constructed by the base block $B_{\ell}$ under the the action of $\GA(\gf(q))$, where
$$\cB= \{g B_{\ell} : g \in \GA(\gf(q))\}.$$
Further, from Theorem \ref{the-con}, it then follows that the incidence structure $\bD(f(x), k):=(\gf(q), \cB_{(f(x),k)})$ is a $2$-$(q, k,\lambda)$ design, where $k$ was defined by Equation (\ref{eqn-kkkk1}) and

\begin{equation}\label{eqn-x}
\lambda=
|\cB_{(f,k)}| \left(
\begin{aligned}
&k  \\
&2  \\
\end{aligned}
\right)
/
\left(
\begin{aligned}
&q\\
&2\\
\end{aligned}
\right)
=k(k-1),
\end{equation}
The proof is then completed.
\end{proof}

%(\uppercase\expandafter{\romannumeral2})  The proof of Theorem \ref{thm-main2} is as follows :
\begin{proof}[Proof of Theorem \ref{thm-main2}]
  The proof is similar to that of Theorem \ref{thm-main2}. By definition, from Lemma \ref{lem-kk1} we have
\begin{eqnarray}\label{eqn-kkkk2}
k=|B_{\ell}|=|\{x^{p^l+1}+x: x \in \gf(q)\}|=q-\frac{pq-1}{2(p+1)}.
\end{eqnarray}
Define the group
\begin{eqnarray}\label{eqn-Ag12}
\GA_1(\gf(q))=\{ux+v: (u, v) \in \gf(q)^* \times \gf(q), u \in QR \}.
\end{eqnarray}
It is clear that the size of the group  $\GA_1(\gf(q))$ is $|\GA_1(\gf(q)|=q(q-1)/2$.
%Furthermore, from Lemma we have $\GA_1(\gf(q))$ is $2$-homogeneous on $\gf(q)$.
The stabilizer of $B_{\ell}$ under $\AG_1(\gf(q))$ is defined by
$$
\Stab(B_{\ell})=\{ux+v: (u,v) \in \gf(q)^* \times \gf(q), u\in \textup{QR}, uB_{\ell}+v=B_{\ell}\},
$$
where $B_{\ell}$ was defined by (\ref{eqn-Je11}).
We then deduce that
\begin{equation}\label{eqn-xx}
|\cB_{(f,k)}|=\frac{|\GA_1(\gf(q)|}{|\Stab(B_{\ell})|} =q(q-1)/2~~~~~~~~~~~%(by~ Lemma ~\ref{lem-miu2}).
\end{equation}
by Lemma \ref{lem-miu2}. By definitions and the result (\uppercase\expandafter{\romannumeral2}) of Lemma \ref{lem-dengzhi}, $(\gf(q), \cB_{(f(x),k)})$ can be seen as $(\gf(q), \cB)$ constructed by the base block $B_{\ell}$ under the the action of $\GA_1(\gf(q))$, where
$$\cB= \{g B_{\ell} : g \in \GA_1(\gf(q))\}.$$
From Theorem \ref{the-con} and Lemma \ref{lem-qici}, it then follows that the incidence structure $(\gf(q), \cB_{(f(x),k)})$ is a $2$-$(q, k, \lambda)$ design, where $k$ was defined by Equation (\ref{eqn-kkkk2}) and

\begin{equation}\label{eqn-x}
\lambda=
|\cB_{(f,k)}| \left(
\begin{aligned}
&k  \\
&2  \\
\end{aligned}
\right)
/
\left(
\begin{aligned}
&q\\
&2\\
\end{aligned}
\right)
=k(k-1)/2,
\end{equation}
The desired conclusion then follows.
\end{proof}

\section{Summary and concluding remarks}\label{sec-summary}

In this paper, based on the general constructions of $t$-designs from polynomials over $\gf(q)$ in \cite{DT192}, quadratic functions were used to
construct $t$-designs. It was shown that infinite families of $2$-designs were produced and their parameters were also explicitly determined. Furthermore, the results in this paper gave an affirmative answer to Conjecture $3$ in  Ding and Tang \cite{DT192} and generalized the result. We remark that this paper does not consider the case that $q$ is an odd prime power with $q\equiv 1 \pmod 4$, since Magma program shows that the corresponding
incidence structures are not $2$-designs.
To conclude this paper, we further presents the following two conjectures, which are the complements of the main results of this paper.
\begin{conj}
Let $p=2$ ,$\ell, m$ be two positive integers with $m\ge 3$, $\frac{m}{4}-1 \le \ell \le  m-1  $ and $\gcd(\ell, m)=1$. Let $f(x)=x^{2^{\ell}+1}$.
Then the incidence structure $\bD(f(x), k):=(\gf(q), \cB_{(f(x),k)})$ is a $2$-$(q, k, k(k-1))$ design, where $k=\frac{2q+(-1)^m}{3}$.
\end{conj}

\begin{conj}
Let $p$ be an odd prime with $p \equiv 3 \pmod 4$ and $m\ge 3$ be odd. Let $\ell$ be a positive integer with  $\frac{m-2}{4} \le \ell \le m-1 $ and $\gcd(\ell, m)=1$. Let
$f(x)=x^{p^{\ell}+1}$.  Then the incidence structure $\bD(f(x), k):=(\gf(q), \cB_{(f(x),k)})$ is a $2$-$(q, k, \frac{k(k-1)}{2})$ design, where $k=q-\frac{pq-1}{2(p+1)}$.
\end{conj}

%Since there may be other way to choose bases blocks supporting 3-designs, a lot of
%work can be done in this direction. Several conjectures were presented in this paper.
%The reader is cordially invited to join the venture into the topic of this paper.

%\section*{Acknowledgements}
%The research of C. Xiang was supported by the National Natural Science Foundation of China (No.11701187) and the PhD Start-up Fund of the Natural Science Foundation of Guangdong %Province of China (No.2017A030310522). The research of X. Ling was supported by the National Natural Science Foundation of China (No.11871058).
%The research of Q. Wang was supported by...............

%The author is very grateful to the reviewers for their comments that improved the quality and presentation of this paper.

\end{document}